\documentclass{amsart}

\usepackage{amssymb}

\newtheorem{theorem}{Theorem}[section]
\newtheorem{lemma}[theorem]{Lemma}
\theoremstyle{definition}

\theoremstyle{remark}

\numberwithin{equation}{section}

\DeclareMathOperator{\grad}{grad}
\DeclareMathOperator{\End}{End}
\DeclareMathOperator{\Tr}{Tr}
\DeclareMathOperator{\dvrg}{div}
\newcommand{\parder}[2]{\frac{\partial #1}{\partial #2}}
\newcommand{\deriv}[1]{\frac{\partial}{\partial #1}}
\newcommand{\tspace}{\mathrm{T}}
\newcommand{\cotspace}{\mathrm{T}^{\ast}}
\newcommand{\reals}{\mathbb{R}}
\newcommand{\covdiff}{d_{\nabla}}
\newcommand{\id}{\mathbf{1}}

\begin{document}

\title{Thermodynamics of Viscous Flows on Surfaces}

\author{A.~Duyunova}
\email{anna.duyunova@yahoo.com}
\address{V. A. Trapeznikov Institute of Control Sciences
    of Russian Academy of Sciences
    65 Profsoyuznaya street, Moscow 117997, Russia}
\author{V.~Lychagin}
\email{valentin.lychagin@uit.no}
\address{V. A. Trapeznikov Institute of Control Sciences
    of Russian Academy of Sciences
    65 Profsoyuznaya street, Moscow 117997, Russia}
\author{S.~Tychkov}
\email{tychkovsn@ipu.ru}
\address{V. A. Trapeznikov Institute of Control Sciences
    of Russian Academy of Sciences
    65 Profsoyuznaya street, Moscow 117997, Russia}

\subjclass[2010]{Primary 76S05, 58J37}

\keywords{thermodynamics, Navier--Stokes equation, co-existence curve}

\begin{abstract}
In this paper, thermodynamics of a moving medium is considered.
Main goal of this paper is to describe thermodynamics state equations
and equations of co-existence manifolds
of media on two- and three-dimensional manifolds.
To this end the phase space of a medium is extended with
the deformation tensor and the stress tensor.
\end{abstract}

\maketitle

\section{Introduction}

This paper we continues our study \cite{dlt_one_2023} of
thermodynamics of a moving medium.
We consider a two- and three-dimensional flows of a medium and introduce
an extended thermodynamic phase space that includes tensor quantities, namely,
deformation $\Delta$ and a stress $\sigma$
together with the ordinary quantities: mass density $\rho$,
entropy density $s$, energy density $e$, temperature $T$ and chemical potential
$\eta$. To analyze this phase space we use the methods, which
are usually applied (cf. \cite{LY3}) to the classical case only.
Note that by a \textit{moving} medium we mean a medium that flows with
a non-zero acceleration.

This paper illustrates the difference between still and moving media
from the thermodynamics standpoint.

Recall \cite{dlt_moving_2024} that the
generalized Navier--Stokes system describing motion of
a fluid on an oriented $n$-dimensional Riemannian manifold $(M, g)$
has the form
\begin{equation}\label{eq:ns}
\left\{
\begin{aligned}
&\rho\left(\parder{X}{t} + \nabla_X X \right)=\dvrg\sigma,\\
&\parder{\rho}{t} + X(\rho) + \rho\dvrg X = 0,\\
&c_p\rho\left(\parder{T}{t} + X(T) \right) = \langle \sigma, \Delta\rangle + \dvrg (\varkappa\grad T)
\end{aligned}
\right.
\end{equation}
where $\rho$ and $e$ are the mass and the energy densities,
$\sigma\in\cotspace M\otimes\tspace M$
and $\Delta\in\tspace M\otimes\cotspace M$
are the stress and the deformation tensor fields,
$X$ is the flow velocity field, $\varkappa$
is a thermal conductivity of the medium,
$c_p$ is a heat capacity at constant pressure.

The deformation tensor $\Delta$ is defined as
$\Delta = \covdiff X\in \tspace M$,
where $\tspace M$ is the tangent space at a point of $M$, and $\nabla$
is the Levi-Civita connection associated with $g$,
and $\covdiff$ is the covariant differential.

Note that if a medium flows with a constant velocity, which means $\Delta=0$,
the usual thermodynamics is applied.

The paper is organized as follows.
In Section \ref{sec:therm} we recall main results regarding
the geometrical approach to thermodynamics.
The notion of a Newton medium is used to determine how the thermodynamic state
depends on the deformation tensor.
In Section \ref{sec:23dim} these constructions are applied
to two- and three-dimensional flows. We calculate a co-existence
curve for phases of moving medium.
In the Sections \ref{sec:surf} and \ref{sec:ns}, we study a special case of flows
on a two-dimensional manifold and give an example of the Navier--Stokes equations for flows on a plane.

\section{Thermodynamics of moving media}\label{sec:therm}
In this section, the necessary for the paper thermodynamics constructions
are recalled.

The thermodynamic state of the fluid in a small neighborhood of a point
is given by the extensive quantities: mass density $\rho$,
internal energy density $e$,
entropy density $s$,
deformation tensor $\Delta$, and the intensive ones:
chemical potential $\eta$,
temperature $T$ and
stress tensor $\sigma$.

Since we consider flows on an oriented Riemannian $(M, g)$,
its thermodynamic phase space is a bundle $\Phi\to M$
with a fiber
\[
\Phi_a = \reals^{5} \times \End \cotspace_a M \times \End \tspace_a M
\]
at a point $a$,
$\Phi_a = (s,\eta,\rho, T, e, \sigma, \Delta)$.
This space is equipped with a contact structure
\[
\alpha = de -T\,dS - \sum\limits_{i,j}\sigma_{ij}\,d\Delta_{ij} - \eta\,d\rho.
\]

The Levi-Civita connection $\nabla$ allows
to compare thermodynamic states at different points of the manifold $M$ with
contact transformations generated by the parallel transports.
This means that the thermodynamics states of the moving medium
is invariant with respect to the Riemannian holonomy group of the manifold $M$.

The thermodynamic state can be given with a single
function $e=e(s, \rho, \Delta)$. But it is more convenient to
introduce the density of Helmholtz free energy
$h = h(\rho,T,\Delta) = e(s, \rho, \Delta) - T s$. In terms of $h$,
we have $\sigma=h_{\Delta}$, $\eta=h_{\rho}$, $e=h-Th_T$.

In terms of the Helmholtz free energy, the pseudo-Riemannian structure $\kappa$
has the form
\begin{equation}\label{eq:kappa}
\begin{aligned}
\kappa=\frac{1}{T}
(&
\frac{\partial^2h}{\partial T^2}\,dT^2-
\frac{\partial^2h}{\partial \rho^2}\,d\rho^2-
\sum\limits_{i,j=1}^n\sum\limits_{k,l=1}^n
\frac{\partial^2h}{\partial \Delta_{ij}\partial \Delta_{kl}}\,
d\Delta_{ij}\cdot d\Delta_{kl} \\
&-2\sum\limits_{i,j=1}^n
\frac{\partial^2h}{\partial \Delta_{ij}\partial \rho}\,
d\Delta_{ij}\cdot d\rho).
\end{aligned}
\end{equation}

We consider Newtonian media, i.~e. media
that satisfy `Hooke's law',
i.~e. the stress tensor $\sigma$ depends on the deformation $\Delta$
linearly. Also, since $\sigma_{ij}=h_{\Delta_{ij}}$, we get that
the Helmholtz free energy density is a quadratic function in $\Delta$.

As we mentioned above, it is
required that the thermodynamic state is
$SO(g)$-invariant \cite{dlt_moving_2024}. Thus
free Helmholtz energy is a function of
$SO(g)$-invariants of the deformation $\Delta$,
\[
h=
\frac{1}{2}\left(
a(\rho,T) \Tr\Delta^2 +
b(\rho,T) \Tr\Delta\Delta^{\prime} +
c(\rho,T) \Tr^2\Delta\right)+
d(\rho,T) \Tr\Delta+
e(\rho,T),
\]
where $\Delta^{\prime}$,
$\Delta^{\prime}_{ij} = \sum\limits_{k,l} g^{-1}_{ik}\Delta_{lk}g_{lj}$,
is the operator adjoint to
$\Delta$, and $a$, $b$, $c$, $d$, $e$ are some functions.
The stress tensor takes the form
\[
\sigma =
a(\rho,T)\Delta^{\prime} +
b(\rho,T)\Delta +
(c(\rho,T)\Tr\Delta + d(\rho,T)) \mathbf{1}.
\]

It is an established assumption \cite{landau1987} that the stress is equal
to the negative thermodynamic pressure $p$ when the deformation vanishes,
i.~e., $p=-d$.

Recall \cite{dlt_moving_2024} that thermodynamic pressure $p$ is
expressed in terms of the Helmholtz free energy
density $h_0(\rho, T)$ for still media as follows,
\[
p=\rho \parder{h_0(\rho, T)}{\rho}- h_0(\rho, T).
\]

\section{Thermodynamics of  two- and three-dimensional flows}\label{sec:23dim}

In this section, the general picture we discussed in Section \ref{sec:therm}
is applied to two- and three-dimensional flows.

Since we consider an isotropic Newtonian medium,
the Helmholtz free energy density is a quadratic function
in $\Delta$. Thus we have
\[
h = \frac 1 2\left(
\left(\mu+\tau\right)\Tr\Delta^2 +
\left(\mu-\tau\right)\Tr\Delta\Delta^{\prime} +
\left(\zeta-\frac 2 n \mu\right)\Tr^2\Delta
\right)
-p\Tr\Delta + h_0.
\]
Here $n=2,3$ is the dimension of the manifold.

Calculating the stress tensor $\sigma=h_{\Delta}$, we get
a well-known linear function in $\Delta$, that is,
\[
\sigma = -p\id + \mu\left(\Delta+\Delta^{\prime}\right)
+ \tau\left(\Delta-\Delta^{\prime}\right)
+ \left(\zeta - \frac 2 n\mu\right)\Tr\Delta\id.
\]
Here $\mu$ is the viscosity, $\zeta$ is the second, or volume,
viscosity, and $\tau$ is the rotational viscosity. These quantities are
considered as functions of the density $\rho$ and the temperature $T$.

For brevity, we introduce the following notation:
\[
d_1 = \Tr\Delta,\quad
d_2 = \Tr\Delta^2,\quad
d_3 = \Tr\Delta\Delta^{\prime}.
\]

Substituting this function into \eqref{eq:kappa}, we
obtain a formula for the pseudo-Riemannian, which is quite
long, so we do not give it here.

Calculating the condition for degeneracy of this form, we
get the following lemma.

\begin{lemma}
Co-existence manifold of a moving Newton medium is defined by
the equation:
\[
\begin{aligned}
&\frac{\partial^2 h}{\partial T^2} \left(\left(d_2+d_3-\frac{2}{n}d_1^2\right)\mu^2\deriv{\rho}
\left(\frac{1}{\mu}\deriv{\rho}\ln\mu\right)
+\left(d_2-d_3\right)\tau^2\deriv{\rho}
\left(\frac{1}{\tau}\deriv{\rho}\ln\tau\right)\right.\\
&\left.+d_1^2\zeta^2\deriv{\rho}
\left(\frac{1}{\zeta}\deriv{\rho}\ln\zeta\right)
-2d_1\left(
\frac{\partial^2p}{\partial\rho^2} +\parder{p}{\rho}\deriv{\rho}\ln\zeta
\right)
+\left(\parder{p}{\rho}\right)^2+2\frac{\partial^2h_0}{\partial\rho^2}\right)=0.
\end{aligned}
\]
\end{lemma}

\section{Thermodynamics of flows on surfaces}\label{sec:surf}

In this section, we study a special case of flows
on a two-dimensional manifold with a metric $g$ and the
associated volume form $\Omega_g$.

Denote the Riemannian holonomy group at a point $a\in M$ as
$H(M, a) \in SO(g)$.
Then the thermodynamic state given by the Helmholtz free energy
has to be a function quadratic in $H(M, a)$-invariants.

Consider the first non-trivial case of the Berger classification \cite{berg}
of irreducible Riemannian holonomy groups,
$H(M, a) = U(1)$, $\dim M= 2$, and
$M$ is a K\"ahler manifold.

In this case, there is a complex structure on $M$, $J\in\End\tspace M$,
i.e., $J^2 = -1$, $J$ is an isometry of the metric $g$, and $J$
is preserved by the parallel transport.

\begin{lemma}
All first and second degree polynomial invariants are generated
by the invariants
\[
t_1 = \Tr\Delta,\qquad
t_2 = \Tr J\Delta,\qquad
t_3 = \Tr \Delta(\Delta +\Delta^{\prime}).
\]
\end{lemma}
\begin{proof}
Let us enumerate all possible first and second order invariants
of the operators $J$, $\Delta$, $\Delta^{\prime}$. According
to the Procesi theorem \cite{procesi2007}, the following invariants
generate all polynomial invariants of order equal or less than two.
\begin{align*}
&t_1 = \Tr\Delta,\quad
t_2 = \Tr J\Delta,\quad
t_3 = \Tr \Delta(\Delta + \Delta^{\prime}),\quad
t_4 = \Tr \Delta(\Delta - \Delta^{\prime}),\\
&t_5 = \Tr J\Delta^2,\quad
t_6 = \Tr J\Delta^{\prime}J\Delta,\quad
t_7 = \Tr (J\Delta)^2.
\end{align*}
In this list, invariants that vanish and those
which coincide up to a sign are omitted.

Straightforward computations show that
\[
t_4=t_2^2,\quad
t_5=t_1 t_2,\quad
t_6=\frac{t_3}{2}-\frac{t_2^2}{2}-t_1^2,\quad
t_7=\frac{t_3}{2}+\frac{t_2^2}{2}-t_1^2.
\]
Note that these relations between invariants $t_1,\dots,t_7$
are true only for flows on surfaces.
\end{proof}

Thus, the density of the free Helmholtz energy for Newtonian media has the form
\begin{align*}
h =& \frac{1}{2}\left(
\mu  \Tr \Delta(\Delta +\Delta^{\prime}) +
\tau \Tr^2 J\Delta +
\alpha (\Tr\Delta)(\Tr J\Delta) +
(\zeta - \mu) \Tr^2\Delta
\right)\\
&- (p \Tr\Delta + q \Tr J\Delta)
+ h_0.
\end{align*}
Note that, in addition to the three viscosities $\mu$, $\tau$, $\zeta$
(see Section \ref{sec:23dim}) and the pressure $p$, we have another viscosity-like
quantity $\alpha$ and a `rotational' pressure $q$.

Substituting this quadratic equation of state into \eqref{eq:kappa} we
get the following lemma.
\begin{lemma}
Co-existence curve of a moving Newton medium on a surface is defined by
the equation:
\[
\begin{aligned}
&h_{TT}\left(\left(\left(\left(t_1^2-t_4 \right) \mu_{\rho\rho}-
\eta_{\rho\rho} t_1^2-\alpha_{\rho\rho} {t_1} {t_2} -
\tau_{\rho\rho}t_2^2+2p_{\rho\rho} {t_1} +2q_{\rho\rho} {t_2}
-2 {h}_{0,\rho\rho}\right) \alpha^{2}\right.\right.\\
&\left.\left.+8 \left(\tau_{\rho} {t_2} +\frac{\alpha_{\rho} {t_1}}{2}
-q_{\rho}\right)\left(\eta_{\rho} {t_1} +\frac{\alpha_{\rho} t_2}{2}- p_{\rho}\right) \alpha +
\left(\left(4\left(t_4-t_1^2 \right) \mu_{\rho\rho}+4 \eta_{\rho\rho} t_1^2\right.\right.\right.\right.\\
&\left.\left.\left.\left.+4 \alpha_{\rho\rho} {t_1} {t_2} +4 \tau_{\rho\rho} t_2^2-8p_{\rho\rho} {t_1} -8 q_{\rho\rho} {t_2} +8 {h}_{0,\rho\rho}\right) \tau-8 \left(\tau_{\rho} {t_2} +\frac{\alpha_{\rho} {t_1}}{2}-q_{\rho}\right)^2\right) \eta \right.\right.\\
&\left.\left.-8 \tau  \left(\eta_{\rho} {t_1} +\frac{\alpha_{\rho} {t_2}}{2}-p_{\rho}\right)^2\right) \mu+8 \left(t_1^2-{t_4} \right) \left(\tau\eta - \frac{\alpha^2}{4}\right) \mu_{\rho}^2\right)=0.
\end{aligned}
\]
\end{lemma}

\section{Navier--Stokes equations on a plane}\label{sec:ns}
As an example, we consider Navier--Stokes equations on a Euclidean plane
($g=dx^2+dy^2$),
given that the thermodynamic state equation $h=h(\rho, T,\Delta)$
is quadratic in
$\Delta$, and viscosities $\mu$, $\tau$, $\zeta$ and $\alpha$ are
constants.

Let $x,y$ be coordianates on the plane,
$u(t,x,y)\partial_x+v(t,x,y)\partial_y$ be the flow velocity field.

Then the stress tensor $\sigma=h_{\Delta}$ is given by
\begin{align*}
\sigma=&(-p+(\mu +\zeta)u_{x}+(\zeta -\mu) v_{y}+\alpha(u_{y}-v_{x}))\,\partial_x\otimes dx\\
+&(-q+(\mu +\tau) u_{y}+(\mu -\tau) v_{x}+
\alpha(u_{x}+v_{y}))\,\partial_x\otimes dy\\
-&(-q-(\mu -\tau) u_{y}-(\mu +\tau)v_x+
\alpha(u_x+ v_y))\,\partial_y\otimes dx\\
+&(-p+(\mu +\zeta)v_{y}+(\zeta -\mu) u_x+\alpha(u_{y}-v_{x}))\,\partial_y\otimes dy
\end{align*}

Substituting the stress tensor into the first equation of \eqref{eq:ns},
we get
\begin{equation*}
\left\{
\begin{aligned}
&\rho(u_t + u u_x+v u_y) + p_x - q_y
-(\mu+\zeta) u_{xx}+(\tau-\mu)u_{yy}-(\tau +\zeta)v_{xy}+
\alpha(v_{xx}+ v_{yy})=0,\\
&\rho(v_t + u v_x + v v_y) + p_y+q_x- (\mu +\zeta) v_{yy}
+(\tau-\mu) v_{xx} -(\tau +\zeta) u_{xy}
-\alpha (u_{xx}+u_{yy})=0.
\end{aligned}
\right.
\end{equation*}
\textbf{Acknowledgments.}
This research was partially supported by Russian
Science Foundation grant number 21-71-20034.

%
%

\end{document}